\documentclass[review]{elsarticle}
\usepackage[letterpaper,top=2cm,bottom=2cm,left=3cm,right=3cm,marginparwidth=1.75cm]{geometry}
\usepackage{lineno,hyperref,mathtools}

\journal{Journal of \LaTeX\ Templates}

\usepackage{amssymb}
\usepackage{bm}
\usepackage{amssymb}
\usepackage{amsthm}
\usepackage{multirow,booktabs}
\usepackage{cases}
\usepackage{threeparttable}

\newtheorem{theorem}{Theorem}

\newtheorem{lemma}[theorem]{Lemma}

\newtheorem{example}[theorem]{Example}

\newcommand{\GRS}{{\mathrm{GRS}}}

\newcommand{\C}{{\mathcal{C}}}
\newcommand{\F}{{\mathbb{F}}}

\begin{document}

\begin{frontmatter}

\title{New MDS self-dual codes over finite field $\F_{r^2}$}
\tnotetext[mytitlenote]{This research work is supported by the National Natural Science Foundation of China under Grant Nos. U21A20428 and 12171134.}

\author[mymainaddress]{Ruhao Wan}
\ead{wanruhao98@163.com}

\author[mymainaddress]{Yang Li}
\ead{yanglimath@163.com}

\author[mymainaddress]{Shixin Zhu\corref{mycorrespondingauthor}}
\cortext[mycorrespondingauthor]{Corresponding author}
\ead{zhushixinmath@hfut.edu.cn}

\address[mymainaddress]{School of Mathematics, HeFei University of Technology, Hefei 230601, China}

\begin{abstract}
    MDS self-dual codes have nice algebraic structures and are uniquely determined by lengths. Recently, the construction of
	MDS self-dual codes of new lengths has become an important and hot issue in coding theory. In this paper, we develop the
	existing theory and construct six new classes of MDS self-dual codes. Together with our constructions, the proportion
	of all known MDS self-dual codes relative to possible MDS self-dual codes generally exceed 57\%. As far as we know, this is the
	largest known ratio. Moreover, some new families of MDS self-orthogonal codes and MDS almost self-dual codes are also constructed.
\end{abstract}


\end{frontmatter}

\section{Introduction}\label{sec1}
Let $q=r^2$, where $r$ is an odd prime power. Denote $\F_q$ be the finite field with $q$ elements. An $[n,k,d]$ linear code
$\C$ over $\F_q$, denoted by $[n,k,d]_q$, can be seen as a $k$-dimensional subspace of $\F_q^n$ with minimum distance $d$. There
are many tradeoffs between $n$, $k$ and $d$. One of the most interesting tradeoffs is so-called the Singleton bound, which yields
$d\leq n-k+1$. When the equation is established, i.e., $d=n-k+1$, $\C$ is called an maximum distance separable (MDS) code. The Euclidean
inner product of two vectors $\mathbf{x}=(x_{1},x_{2},...,x_{n})$ and $\mathbf{y}=(y_{1},y_{2},...,y_{n})$ in $\F_{q}^{n}$ is defined
by $$\mathbf{x}\cdot\mathbf{y}=\sum _{i=1}^{n} x_{i}y_{i}.$$ With respect to Euclidean inner product, we define the dual code of
$\C$ as
\begin{center}
$\C^{\bot}=\{\mathbf{x}\in  \F_{q}^{n}:\ \mathbf{x}\cdot\mathbf{c}=0,\ \forall\ \mathbf{c}\in \C\}.$
\end{center}
If $\C\subseteq \C^{\bot}$, we call $\C$ (Euclidean) self-orthogonal. If $\C=\C^{\bot}$, we call $\C$ (Euclidean) self-dual.
In particular, if $\C$ is MDS and $\C=\C^{\bot}$, $\C$ is called an MDS (Euclidean) self-dual code. Clearly, the dimension $k$ and minimal
distance $d$ of an MDS self-dual code are uniquely determined by its length $n$. Furthermore, the length of a self-dual code can only be even.
A relatively weak case is almost self-dual. Usually, we call $\C$ an almost self-dual code if $\C \subseteq \C^{\bot}$ and
${\rm dim}(\C^{\bot})={\rm dim}(\C)+1$, where $\dim(\C)$ denotes the dimension of $\C$. Some constructions of almost self-dual codes were given
and the reader may refer to \cite{RefJ18,RefJ20}.

Due to the nice algebraic structures, MDS self-dual codes have wide and important applications in distributed storage systems \cite{RefJ2},
linear secret sharing schemes \cite{RefJ1,RefJ5} and unimodular integer lattices \cite{RefJ xin}. What also needs to be emphasized is
that some outstanding works were given in \cite{RefJ11,RefJ12,RefJ17}, which greatly promoted the research enthusiasm for the construction of MDS
self-dual codes. Of course, as we pointed above, all parameters of an MDS self-dual code are completely determined by its length. Therefore,
the main problem of the construction of MDS self-dual codes is which length a $q$-ary MDS self-dual code exists. For the case even $q$, the problem
was completely solved in \cite{RefJ3}. It means that for constructing new MDS self-dual codes, it is sufficient to consider the case where
$q$ is an odd prime power. For this topic, there has been a lot of results based on different methods and we can summarize them as follows.

%Furthermore, GRS codes, as a special class of MDS codes, have found wide applications in practice, such as compact disc players, disk drives, satellite communications and quantum MDS codes (see \cite{RefJ8}).

In \cite{RefJ7}, a criterion for generalized Reed-Solomon (GRS) codes being MDS self-dual codes were introduced by Jin et al. and were developed by
Fang et al. and Yan to extended GRS codes in \cite{RefJ4,RefJ10}. After this, Zhang et al. proposed a unified method to construct MDS self-dual
codes via GRS codes. By selecting suitable evaluation sets of (extended) GRS codes, MDS self-dual codes can also be explictly derived in some ways.
Specifically, in \cite{RefJ19} the authors considered the evaluation set to be a subgroup of some finite field, where its coset belongs to a bigger subgroup.
In \cite{RefJ18}, Fang considered the union of two disjoint multiplicative subgroups and took their cosets as the evaluation sets.
Then the conclusions in \cite{RefJ18} were further generalized for new families of MDS self-dual codes with flexible paramenters in \cite{RefJ20}.
And in \cite{RefJ16}, the authors constructed the evaluation set via two multiplicative subgroups with nonempty intersections and their cosets.
In addition to the methods mentioned above, the approaches of constructing MDS self-dual codes through linear subspace and long MDS
self-dual codes from short codes can be found in \cite{RefJ6} and \cite{RefJ9}, respectively. In general, no matter which method is used, a large number
of MDS self-dual codes can be constructed. In Table \ref{tab:2}, we list some known constructions.

%Although people have done a lot of work, it is still an unsolved problem to construct MDS self-dual codes with all possible lengths.

\newcommand{\tabincell}[2]{\begin{tabular}{@{}#1@{}}#2\end{tabular}}
\begin{table}
\caption{Some known classes of $q$-ary MDS self-dual codes of even length $n$, where $q=r^2$}
\label{tab:2}
\begin{center}
\resizebox{\textwidth}{105mm}{
	\begin{tabular}{ccc}
		\hline
		Class  & $n$ even & References\\
		\hline
		1   &  $2\leq n \leq 2r$& \cite{RefJ13}  \\
		2   & $n<q-1$, $n\mid (q-1)$ or $(n-1)\mid (q-1)$ or $(n-2)\mid (q-1)$& \cite{RefJ10}  \\
		3   & $n=tr$, $1\leq t\leq r$, $t$ even & \cite{RefJ10}  \\
		4   & $n=tr+1$, $1\leq t\leq r$, $t$ odd & \cite{RefJ10}  \\
		5   & $4^nn^2\leq q$ & \cite{RefJ7}  \\
		6   & $n=tm$,   $tm$ even,  $1\leq t\leq \frac{r-1}{\gcd(r-1,m)}$, $\frac{q-1}{m}$ even       &   \cite{RefJ14}  \\
		7   & $n=tm+1$, $tm$ odd, $2\leq t\leq \frac{r-1}{\gcd(r-1,m)}$ and $m\mid q-1$ & \cite{RefJ14}  \\
		8   & $n=tm+2$, $tm$ even, $2\leq t\leq \frac{r-1}{\gcd(r-1,m)}$ and $m\mid q-1$&  \cite{RefJ14}  \\
		9   & $n=tm$,   $tm$ even,   $1\leq t\leq \frac{r+1}{\gcd(r+1,m)}$, $\frac{q-1}{m}$ even       & \cite{RefJ15}\\
	    10  & $n=tm+1$, $tm$ odd, $2\leq t\leq \frac{r+1}{2\gcd(r+1,m)}$ and $m\mid q-1$& \cite{RefJ15}\\
		11  & \tabincell{c}{$n=tm+2$, $tm$ even(except when $t$ is even, $m$ is even \\and $r\equiv 1($mod$\,4)$), $1\leq t\leq \frac{r+1}{\gcd(r+1,m)}$ and $m\mid q-1$}& \cite{RefJ15}\\
		12  & $n=tm$, $1\leq t\leq \frac{s(r-1)}{\gcd(s(r-1),m)}$, $s$ even, $s\mid m$, $\frac{r+1}{s}$ and $\frac{q-1}{m}$ even& \cite{RefJ15}\\
		13  & $n=tm+2$, $1\leq t\leq \frac{s(r-1)}{\gcd(s(r-1),m)}$, $s$ even, $s\mid m$, $\frac{r+1}{s}$ and $\frac{q-1}{m}$ & \cite{RefJ15}\\
		14  & $n=tm$, $1\leq t\leq \frac{r-1}{n_2}$, $m\mid q-1$, $n_2=\frac{r+1}{\gcd(r+1,m)}$ even & \cite{RefJ19}\\
		15  & $n=tm+2$, $1\leq t\leq \frac{r-1}{n_2}$, $m\mid q-1$, $n_2=\frac{r+1}{\gcd(r+1,m)}$ & \cite{RefJ19}\\
		16  & \tabincell{c}{$n=tm$, $1\leq t\leq \frac{r+1}{n_2}$, $m\mid q-1$ even, $n_1=\gcd(r-1,m)$, \\$n_2=\frac{r-1}{n_1}$ and $\frac{r-1}{n_1}+tn_2$ even }& \cite{RefJ19}\\
		17  & \tabincell{c}{$n=tm+2$, $1\leq t\leq \frac{r+1}{n_2}-1$, $m\mid q-1$, $n_2=\frac{r-1}{\gcd(r-1,m)}$, \\$n_2$ even and $\frac{(r+1)(t-1)}{2}$ even, or $n_2$ odd and $t$ even} & \cite{RefJ19} \\
		18  & $n=s(r-1)+t(r+1)$, $s$ even, $1\leq s\leq \frac{r+1}{2}$, $1\leq t\leq \frac{r-1}{2}$, $r\equiv 1({\rm mod}~4)$ & \cite{RefJ18}\\
        19  & $n=s(r-1)+t(r+1)$, $s$ odd, $1\leq s\leq \frac{r+1}{2}$, $1\leq t\leq \frac{r-1}{2}$, $r\equiv 3({\rm mod}~4)$ & \cite{RefJ18}\\
        20  & \tabincell{c}{$n=s\frac{q-1}{a}+t\frac{q-1}{b}$, $1\leq s\leq \frac{a}{\gcd(a,b)}$, $1\leq t\leq \frac{b}{\gcd(a,b)}$,
        $2a\mid b(r+1)$, \\ $2b\mid a(r-1)$, $a\equiv 2({\rm mod}~4)$, $b$ even, $s$ even and $r\equiv 1({\rm mod}~4)$}& \cite{RefJ20}\\
        21  & \tabincell{c}{$n=s\frac{q-1}{a}+t\frac{q-1}{b}$, $1\leq s\leq \frac{a}{\gcd(a,b)}$, $1\leq t\leq \frac{b}{\gcd(a,b)}$,
        $2a\mid b(r+1)$, \\ $2b\mid a(r-1)$, $b\equiv 2({\rm mod}~4)$, $a$ even, $\frac{(r+1)b}{2a}s^2$ odd and $r\equiv 3({\rm mod}~4)$}& \cite{RefJ20}\\
        22 &  \tabincell{c}{$n=s\frac{q-1}{a}+t\frac{q-1}{b}+2$, $1\leq s\leq \frac{a}{\gcd(a,b)}$, $1\leq t\leq \frac{b}{\gcd(a,b)}$,
        $2a\mid b(r+1)$,\\ $2b\mid a(r-1)$, $a\equiv 2({\rm mod}~4)$, $b$ even, $s$ odd and $r\equiv 1({\rm mod}~4)$}& \cite{RefJ20}\\
        23 &  \tabincell{c}{$n=s\frac{q-1}{a}+t\frac{q-1}{b}+2$, $1\leq s\leq \frac{a}{\gcd(a,b)}$, $1\leq t\leq \frac{b}{\gcd(a,b)}$,
        $2a\mid b(r+1)$, \\ $2b\mid a(r-1)$, $b\equiv 2({\rm mod}~4)$, $a$ even, $\frac{(r+1)b}{2a}s^2$ even and $r\equiv 3({\rm mod}~4)$}& \cite{RefJ20}\\
        24 &  \tabincell{c}{$n=s\frac{q-1}{\mu}+t\frac{q-1}{\nu}-\frac{2(q-1)\gcd(\mu,\nu)}{\mu\nu}st$, $1\leq s\leq \frac{\mu}{\gcd(\mu,\nu)}$, $1\leq t\leq \frac{\nu}{\gcd(\mu,\nu)}$, \\
        $\mu\mid \nu(r+1)$, $\nu\mid \mu(r-1)$, $n$, $\mu$, $\frac{(r+1)\nu}{\mu}s+\nu$ even, $r\equiv 3({\rm mod}~4)$} & \cite{RefJ16}\\
        25 &  \tabincell{c}{$n=s\frac{q-1}{\mu}+t\frac{q-1}{\nu}-\frac{2(q-1)\gcd(\mu,\nu)}{\mu\nu}st+1$, $1\leq s\leq \frac{\mu}{\gcd(\mu,\nu)}$, $1\leq t\leq \frac{\nu}{\gcd(\mu,\nu)}$, \\
        $\mu\mid \nu(r+1)$, $\nu\mid \mu(r-1)$, $n$ odd, $\mu$, $\frac{(r+1)\nu}{\mu}s+\nu$ and $\frac{\nu(r+1)}{\mu}\frac{s(s-1)}{2}$ even, $r\equiv 3({\rm mod}~4)$}& \cite{RefJ16}\\
        26 &  \tabincell{c}{$n=s\frac{q-1}{\mu}+t\frac{q-1}{\nu}-\frac{2(q-1)\gcd(\mu,\nu)}{\mu\nu}st+2$, $1\leq s\leq \frac{\mu}{\gcd(\mu,\nu)}$, $1\leq t\leq \frac{\nu}{\gcd(\mu,\nu)}$, \\
        $\mu\mid \nu(r+1)$, $\nu\mid \mu(r-1)$, $n$, $\frac{(r+1)\nu}{\mu}s$ and $\frac{\nu(r+1)}{\mu}\frac{s(s-1)}{2}$ even, $r\equiv 3({\rm mod}~4)$}& \cite{RefJ16}\\
		\hline
	\end{tabular}}
    \begin{tablenotes}
     \footnotesize
    \item Note: in \cite{RefJ16}, many MDS self-dual codes over $\F_q$ with $q=r^2$ were also introduced.
    \end{tablenotes}
\end{center}
\end{table}

\begin{table}
\caption{Our results}
\label{tab:3}
\begin{center}
\resizebox{\textwidth}{21mm}{
	\begin{tabular}{cc}
		\hline
		 $n$ even & References\\
		\hline
		  $n=sf_1+tf_2$, $2e_2\mid e_1(r-1)$, $e_1\mid e_2(r+1)$ and $4\mid (s-1)(r+1)$  & Theorem \ref{th 2}(1)\\
		  $n=sf_1+tf_2+1$, $2e_2\mid e_1(r-1)$, $e_1\mid e_2(r+1)$ and $4\mid (s-1)(r+1)$  & Theorem \ref{th 2}(2)\\
		  $n=sf_1+tf_2+2$, $2e_2\mid e_1(r-1)$, $e_1\mid e_2(r+1)$ and $4\mid (s-1)(r+1)$ & Theorem \ref{th 2}(3)\\
		  $n=sf_1+tf_2$, $2e_2\mid e_1(r+1)$, $e_1\mid e_2(r-1)$ and $\frac{r+1}{2}(\frac{te_1}{e_2}+1)$ even & Theorem\ref{th 3}(1)\\
		  $n=sf_1+tf_2+1$, $2e_2\mid e_1(r+1)$, $e_1\mid e_2(r-1)$, $\frac{r+1}{2}(\frac{te_1}{e_2}+t)$ even and $4\mid (t-1)(r+1)$ & Theorem \ref{th 3}(2)\\
		  $n=sf_1+tf_2+2$, $2e_2\mid e_1(r+1)$, $e_1\mid e_2(r-1)$, $\frac{r+1}{2}(\frac{te_1}{e_2}+t)$ even and $4\mid (t-1)(r+1)$ & Theorem \ref{th 3}(3)\\
		\hline
	\end{tabular}}
    \begin{tablenotes}
     \footnotesize
    \item In this table: $q=r^2$, $q-1=e_1f_1=e_2f_2$, $e_1\equiv 2^l({\rm mod}~2^{l+1})$, $2^l\mid e_2$, where $l\geq 2$, $D_1=\frac{e_1}{\gcd(e_1,e_2)}$, $D_2=\frac{e_2}{\gcd(e_1,e_2)}$, $1\leq s\leq D_1$ and $1\leq t\leq D_2$.
    \end{tablenotes}
\end{center}
\end{table}

Inspired and motivated by these works, especially the idea introduced in \cite{RefJ20}, we study MDS self-dual codes of new lengths in this paper.
As a result, several new classes of MDS self-dual codes are constructed and listed in Table \ref{tab:3}.
Compared with Table \ref{tab:2}, our MDS self-dual codes have more flexible parameters.

The rest of this paper is organized as follows.
In Section \ref{sec2}, we briefly introduce some basic notations and results on (extended) GRS codes and MDS self-dual codes.
In Section \ref{sec3}, the constructions of new MDS self-dual codes via (extended) GRS codes are discussed.
Detailed comparisons of our conclusions and previous results, including some concrete examples, are given in Section \ref{sec4}.
Moreover, some new families of MDS self-orthogonal codes and MDS almost self-dual codes are presented in Section \ref{sec5}.
Finally, Section \ref{sec6} concludes this paper.

\section{Preliminaries}\label{sec2}

We now recall some basic properties and important results about (extended) GRS codes. Throughout this paper, we always assume that $q=r^2$,
where $r$ is an odd prime power.

Let $\F_{q}$ be the finite field with $q$ elements.
Choose $\mathbf{a}=(\alpha_{1},\alpha_{2},\cdots,\alpha_{n})$
to be an $n$-tuple of distinct elements of $\F_q$.
Put $\mathbf{v}=(v_{1},v_{2},\cdots,v_{n})$ with $v_i\in\F_q^*$.
For an integer $k$ satisfying $0\leq k\leq n\leq q$, the generalized
Reed-Solomon (GRS) code associated to $\mathbf{a}$ and $\mathbf{v}$ is defined by
\begin{equation}\label{eq1}
\begin{split}
GRS_{k}(\mathbf{a},\mathbf{v})=\{(v_{1}f(\alpha_{1}),v_{2}f(\alpha_{1}),\cdots,v_{n}f(\alpha_{n}):\ f(x)\in \F_{q}[x],\ \deg(f(x))\leq k-1\}.
\end{split}
\end{equation}
It is well known that $GRS_{k}(\mathbf{a},\mathbf{v})$ is an $[n,k,n-k+1]_q$ MDS code. And its dual code is also MDS.

Let $\eta(x)$ be the quadratic character of $\F_q^*$ and $\mathbf{a}=(\alpha_{1},\alpha_{2},\cdots,\alpha_{n})\in \F_q^n$ be the same as before.
%Let $\mathbf{a}=\{\alpha_{1},\alpha_{2},...,\alpha_{n}\}\subseteq F_q^n$ be as above.
Denote
\begin{equation}
f_{\mathbf{a}}(x)=\prod_{\alpha\in \mathbf{a}}(x-\alpha) \quad {\rm{and}} \quad L_{\mathbf{a}}(\alpha_i)=\prod_{1\leq j\leq n,j\neq i}(\alpha_i-\alpha_j),
\end{equation}
which will be used frequently in this paper.
In the same way, the above definitions can be extended to any subset of $\F_q$, whose elements are different from each other.
For any polynomial $f(x)=\sum a_ix^i\in \F_q[x]$, the derivative of $f(x)$ is $f'(x)=\sum ia_ix^{i-1}\in \F_q[x]$.
Then for any $\alpha_i\ (1\leq i\leq n)$, we have $L_{\mathbf{a}}(\alpha_i)=f'_{\mathbf{a}}(\alpha_i)$.

The following lemma provides a criterion for a GRS code being self-dual.
\begin{lemma}(\cite{RefJ7})\label{rem1}
Let notations be the same as before. For an even integer $n$ and $k=\frac{n}{2}$, if $\eta(L_{\mathbf{a}}(\alpha_{i}))$ are the same for all $1\leq i \leq n$,
then $GRS_k(\mathbf{a},\mathbf{v})$ defined in Eq. (\ref{eq1}) is an MDS self-dual code.
\end{lemma}

%Let's now introduce some basic notations and results on extended GRS codes.
In coding theory, the $k$-dimensional extended GRS codes of length $n$ associated to $\mathbf{a}$ and $\mathbf{v}$ is defined by
\begin{equation}\label{eq2}
GRS_{k}(\mathbf{a},\mathbf{v},\infty)=\{(v_{1}f(\alpha_{1}),v_{2}f(\alpha_{2}),\cdots,v_{n-1}f(\alpha_{n-1}),f_{k-1}):\ f(x)\in \F_{q}[x],\ \deg(f(x))\leq k-1\},
\end{equation}
where $f_{k-1}$ is the coefficient of $x^{k-1}$ in $f(x)$. It is easy to verify that $GRS_{k}(\mathbf{a},\mathbf{v},\infty)$ is an $[n,k,n-k+1]_q$ MDS code.
And its dual code is also MDS.

The following lemma provides a similar criterion for an extended GRS code being self-dual.
\begin{lemma}(\cite{RefJ10})\label{lem2}
Let notations be the same as before. For an odd integer $n$ and $k=\frac{n+1}{2}$, if $\eta(-L_{\mathbf{a}}(\alpha_i))=1$ for all $1\leq i \leq n$, then $\GRS_k(\mathbf{a},\mathbf{v},\infty)$ defined in
Eq. (\ref{eq2}) is an MDS self-dual code with $v_i^2=(-L_\mathbf{a}(\alpha_i))^{-1}$ for $1\leq i \leq n$.
\end{lemma}

Apart from self-duality, self-orthogonality is also a problem worthy of study. In particular, for MDS self-orthogonal codes, two criterions about (extended) GRS codes were shown in \cite{RefJ18}.
We rephrase them in the following two lemmas.

\begin{lemma}(\cite{RefJ18})\label{lem self-orth 1}
Assume $1\leq k\leq \lfloor \frac{n}{2}\rfloor$. Then $GRS_k(\mathbf{a},\mathbf{v})$ is self-orthogonal if and only if
there exists a nonzero polynomial $\omega(x)=\omega_{n-2k}x^{n-2k}+\dots+\omega_1x+\omega_0\in \F_q[x]$ such that
$\eta(\omega(\alpha_i)L_{\mathbf{a}}(\alpha_i))=1$ for $1\leq i\leq n$.
\end{lemma}

\begin{lemma}(\cite{RefJ18})\label{lem self-orth 2}
Assume $1\leq k\leq \lfloor \frac{n+1}{2}\rfloor$. Then $GRS_k(\mathbf{a},\mathbf{v},\infty)$ is self-orthogonal if and only if
there exists a nonzero polynomial $\omega(x)=-x^{n-2k+1}+\omega_{n-2k}x^{n-2k}+\dots+\omega_0\in \F_q[x]$ such that
 $\eta(\omega(\alpha_i)L_{\mathbf{a}}(\alpha_i))=1$ for $1\leq i\leq n$.
\end{lemma}

In addition, for our constructions, we still need some other important properties. We list them in the followings.
\begin{lemma}(\cite{RefJ13})\label{lem5}
$\rm (1)$ Let $S_1$ and $S_2$ be disjoint subsets of $\F_q$, $S=S_1\cup S_2$.
Then for $b\in S$,
\[ L_{S}(b)=\begin{cases}
L_{S_1}(b)f_{S_2}(b),\quad if\  b\in S_1 \\
L_{S_2}(b)f_{S_1}(b),\quad  if\  b\in S_2.
\end{cases}\]\\
$\rm (2)$ Let $\theta$ be a primitive element of $\F_q$ and $q-1=ef$.
Denote $H=\langle\theta^{e}\rangle$, then for $a\in \theta^i H$, we have
\begin{equation}\label{eq111}
f_{\theta^i H}(x)=x^{f}-\theta^{if} \quad {\rm{and}} \quad L_{\theta^i H}(a)=f'_{\theta^i H}(a)=fa^{f-1}.
\end{equation}
\end{lemma}

\begin{lemma}(\cite{RefJ13})\label{lem xin}
Let $\theta$ be a primitive element of $\F_q$ and $q-1=ef$.
Denote $H=\langle\theta^{e}\rangle$.
Let $S_i=\xi_iH$ $(1\leq i\leq t)$ be $t$ distinct cosets of $H$ in $\F_q^*$ $(0\leq t\leq e-1)$,
$S=\cup_{i=1}^t S_i$, $|S|=tf$.
Then
\begin{equation}
f_S(x)=\prod_{i=1}^{t}\prod_{j=0}^{f}(x-\xi_i\theta^{ej})=\prod_{i=1}^{t}(x^{f}-\xi_i^f)=g(x^f),
\end{equation}
where $g(x)=\prod_{i=1}^t(x-\xi_i^f)=f_{S'}(x)$, $S'=\{\xi_i^f:\ 1\leq i \leq t\}$.
\end{lemma}

\begin{lemma}(\cite{RefJ20})\label{lem6}
Assume that $n\leq q-1$ is even.
Let $\mathbf{a}=(\alpha_1,\alpha_2,\cdots,\alpha_n)\in \F_q^n$ such that
\begin{equation}
\eta(-\alpha_iL_{\mathbf{a}}(\alpha_i))=\eta(-\prod_{i=1}^{n}\alpha_i)=1,\ 1\leq i\leq n.
\end{equation}
Then there exists a $q$-ary MDS self-dual code of length $n+2$.
\end{lemma}

\section{Main results}\label{sec3}
In this section, we construct several new classes of MDS self-dual codes over $\F_q$.
%The main idea of our constructions is to choose suitable evaluation sets such that the corresponding (extended) GRS codes are self-dual.
Let $\theta$ be a primitive element of $\F_q$. Now we consider the two decompositions of $q-1$. For brevity, we emphasize the following notations:
\begin{itemize}
    \item $q-1=e_1f_1=e_2f_2$.
    \item $\alpha=\theta^{e_1}$, $\beta=\theta^{e_2}$, $A=\langle\alpha\rangle$, $B=\langle\beta\rangle$.
    \item $D_1=\frac{e_1}{\gcd(e_1,e_2)}=\frac{f_2}{\gcd(f_1,f_2)}$, $D_2=\frac{e_2}{\gcd(e_1,e_2)}=\frac{f_1}{\gcd(f_1,f_2)}$.
\end{itemize}
\iffalse
 $\bullet $ $q-1=e_1f_1=e_2f_2$.\\
 $\bullet $ $\alpha=\theta^{e_1}$, $\beta=\theta^{e_2}$, $A=\langle\alpha\rangle$, $B=\langle\beta\rangle$.\\
 $\bullet $ $D_1=\frac{e_1}{\gcd(e_1,e_2)}=\frac{f_2}{\gcd(f_1,f_2)}$, $D_2=\frac{e_2}{\gcd(e_1,e_2)}=\frac{f_1}{\gcd(f_1,f_2)}$.
 \fi

Based on notations above, we gave a necessary and sufficient condition for cosets to be different from each other in \cite{RefJ21}. We recall it as Lemma
\ref{mod butong} here, which is extremely crucial for our constructions.

\begin{lemma}(\cite{RefJ21})\label{mod butong}
Let $\{\beta^{i_1},\beta^{i_2},\cdots,\beta^{i_t}\}$ be a subset of $B$, where $i_1,i_2,\dots,i_t$ are distinct module $f_2$.
Then $\beta^{i_\lambda}A\ (1\leq \lambda \leq t)$ are distinct cosets of $A$ in $\F_q^*$ if and only if $i_1, i_2,\dots, i_t$ are
distinct module $D_1$.
\end{lemma}

Now, we present our first construction as follows, which can be used to obtain three classes of MDS self-dual codes.
\begin{theorem}\label{th 2}
Let $q=r^2$, where $r$ is an odd prime power.
Let $e_1\equiv 2^l({\rm mod}~2^{l+1})$ and $2^l\mid e_2$, where $l\geq 2$.
Suppose $2e_2\mid e_1(r-1)$ and $e_1\mid e_2(r+1)$.
Let $1\leq s \leq D_1$, $1\leq t \leq D_2$, $4\mid (s-1)(r+1)$ and $n_1=sf_1+tf_2$.
The following statements hold.
\begin{itemize}
\item[(1)] If $n_1$ is even, there exists a $q$-ary MDS self-dual code of length $n_1$.
\item[(2)] If $n_1$ is odd, there exists a $q$-ary MDS self-dual code of length $n_1+1$.
\item[(3)] If $n_1$ is even, there exists a $q$-ary MDS self-dual code of length $n_1+2$.
\end{itemize}
\end{theorem}

\begin{proof}
(1) Let $\theta$ be a primitive element of $\F_q$ and $\gamma=\theta^{\frac{e_1}{2}}$.
Denote
\begin{equation}\label{eqqq1}
	   M=\bigcup_{\mu=1}^{s}\beta^{i_\mu}A,\  N=\bigcup_{\nu=1}^{t}\gamma^{2j_\nu+1}B\  {\rm{and}}\ S=M \cup N,
\end{equation}
where $i_1, i_2,\dots, i_s$ are distinct module $D_1$ and $j_1, j_2,\dots, j_t$ are distinct module $D_2$.
By Lemma \ref{mod butong},
$\beta^{i_\mu}A\ (1\leq \mu \leq s)$, $\gamma^{2j_\nu+1}B\ (1\leq \nu\leq t)$ are distinct coests of $A$, $B$, respectively.

Now we prove that $M\cap N=\emptyset$ by the method of disproof.
Let $e\in M\cap N$,
then there exists some integers $0\leq \mu_1\leq s-1$, $1\leq \nu_1\leq f_1$, $0\leq \mu_2\leq t-1$ and $1\leq \nu_2\leq f_2$
such that
\begin{equation}\label{equ jiao}
e=\beta^{\mu_1}\alpha^{\nu_1}=\gamma^{2\mu_2+1}\beta^{\nu_2}.
\end{equation}
Note that Eq. (\ref{equ jiao}) is equivalent to
$$\theta^{e_2(\mu_1-\nu_2)+e_1(\nu_1-\mu_2)-\frac{e_1}{2}}=1.$$
It follows that $e_2\mid e_1(\nu_1-\mu_2-\frac{1}{2})$, which implies $e_2\mid \frac{e_1}{2}(2\nu_1-2\mu_2-1)$, i.e.,
$\frac{e_2}{\gcd(e_2,\frac{e_1}{2})}\mid (2\nu_1-2\mu_2-1)$.
Since $2^l\nmid \frac{e_1}{2}$ and $2^l\mid e_2$, we can conclude that $2^l\nmid \gcd(e_2,\frac{e_1}{2})$ and $2\mid \frac{e_2}{\gcd(e_2,\frac{e_1}{2})}$,
which implies $2\mid (2\nu_1-2\mu_2-1)$. Clearly, this is a contradiction.
Hence, $|S| =n_1=sf_1+tf_2$. Denote $f_S(x)=f_M(x)f_N(x)$. By Lemma \ref{lem5}, it is easy to see that
 \[\begin{split}
	 L_S(x)=f'_S(x)=f'_M(x)f_N(x)+f_M(x)f'_N(x).
\end{split}\]

On one hand, let $\beta^{i_\lambda}\alpha^j\in M$ for some $1\leq \lambda \leq s$ and $1\leq j\leq f_1$, then
\[\begin{split}
	 L_S(\beta^{i_\lambda}\alpha^j)&=L_M(\beta^{i_\lambda}\alpha^j)f_N(\beta^{i_\lambda}\alpha^j)\\
                &=f_1\theta^{e_2i_\lambda(f_1-1)}\theta^{-je_1}L_\mathbf{a}(\beta^{i_\lambda f_1})\prod_{h=1}^{t}(\alpha^{f_2j}-\gamma^{(2j_h+1)f_2})\\
                &=f_1\theta^{e_2i_\lambda(f_1-1)}\theta^{-je_1}L_\mathbf{a}(\beta^{i_\lambda f_1})\prod_{h=1}^{t}(\theta^{e_1f_2j}-\theta^{\frac{e_1}{2}(2j_h+1)f_2}),\\
	\end{split}\]
where $\mathbf{a}=(\beta^{i_1f_1},\beta^{i_2f_1},\cdots,\beta^{i_{s}f_1})$.
We can also further derive that $\theta^{\frac{e_1f_2}{2}}\in \F_r^*$ and $$\eta(\prod_{h=1}^{t}(\theta^{e_1f_2j}-\theta^{\frac{e_1}{2}(2j_h+1)f_2}))=1$$
from the fact $\frac{\frac{e_1f_2}{2}}{1+r}=\frac{e_1(r-1)}{2e_2}$.
Note that $\eta(f_1\theta^{e_2i_\lambda(f_1-1)}\theta^{-je_1})=1$, thus, it is sufficient to consider $L_\mathbf{a}(\beta^{i_\lambda f_1})$.
Since $$(\beta^{f_1})^r=\theta^{\frac{e_2(r+1)(r^2-r)}{e_1}}=\beta^{-f_1},$$
it follows that
\[\begin{split}
	        (L_{\mathbf{a}}(\beta^{i_\lambda f_1}))^r &= \prod_{\mu=1 ,\mu\neq \lambda}^{s}(\beta^{-i_\lambda f_1}-\beta^{-i_\mu f_1}) \\
                                              &= \prod_{\mu=1 ,\mu\neq \lambda}^{s}\frac{\beta^{i_\mu f_1}-\beta^{i_\lambda f_1}}{\beta^{(i_\lambda+i_\mu)f_1}}\\
                                              &= (-1)^{s-1}\beta^{-i_\lambda f_1(s-2)-f_1 I} L_{\mathbf{a}}(\beta^{i_\lambda f_1}),\\
\end{split}\]
where $I=\sum_{\mu=1}^s i_\mu$. Moreover, we can know that
\[\begin{split}
(L_\mathbf{a}(\beta^{i_\lambda f_1}))^{r-1}=\theta^{\frac{(s-1)(r^2-1)}{2}-e_2f_1[i_\lambda (s-2)+I]}\\
	\end{split}\]
and
\[\begin{split}
L_\mathbf{a}(\beta^{i_\lambda f_1})=\theta^{\frac{(s-1)(r+1)}{2}-\frac{e_2(r+1)}{e_1}[i_\lambda (s-2)+I]+(r+1)c},\\
	\end{split}\]
where $c$ is some integer.
Considering $2^{l+1}\mid e_2(r+1)$ and $2^{l+1}\nmid e_1$, it follows that $\frac{e_2(r+1)}{e_1}$ is even. Therefore,
\begin{equation}\label{eqqq2}
\eta(L_S(\beta^{i_\lambda}\alpha^j))=\eta(L_\mathbf{a}(\beta^{i_\lambda f_1}))=(-1)^{\frac{(s-1)(r+1)}{2}}.
\end{equation}

On the other hand, let $\gamma^{2j_\delta+1}\beta^k\in N$ for some $1\leq \delta \leq t$ and $1\leq k\leq f_2$.
By Eq. (\ref{eq111}), we have
 \[\begin{split} L_S(\gamma^{2j_\delta+1}\beta^k)&=L_N(\gamma^{2j_\delta+1}\beta^k)f_M(\gamma^{2j_\delta+1}\beta^k)\\
                &=f_2\theta^{\frac{e_1}{2}(2j_\delta+1)(f_2-1)}\theta^{-ke_2}\prod_{l=1,l\neq \delta}^t(\gamma^{(2j_\delta+1)f_2}-\gamma^{(2j_l+1)f_2})(-1)^s\prod_{h=1}^{s}(\beta^{f_1k}+\beta^{f_1i_h})\\
                &=f_2\theta^{\frac{e_1}{2}(2j_\delta+1)(f_2-1)}\theta^{-ke_2}\theta^{\frac{(t-1)e_1f_2}{2}}L_\mathbf{b}(\alpha^{j_\delta f_2})(-1)^s\prod_{h=1}^{s}(\theta^{e_2f_1k}+\theta^{e_2f_1i_h}),\\
	\end{split}\]
where $\mathbf{b}=(\alpha^{j_1f_2},\alpha^{j_2f_2},\cdots,\alpha^{j_{t}f_2})$.
Since $2e_2\mid e_1(r-1)$, then $\alpha^{f_2}\in \F_r^*$, it deduces that $L_\mathbf{b}(\alpha^{j_\delta f_2})\in \F_r^*$.
A similar calculation as above means that
\[\begin{split}
\eta(f_2\theta^{\frac{e_1}{2}(2j_\delta+1)(f_2-1)}\theta^{-ke_2}\theta^{\frac{(t-1)e_1f_2}{2}}L_\mathbf{b}(\alpha^{j_\delta f_2})(-1)^s)=1.
	\end{split}\]
Hence, it is sufficient to consider %$\omega\triangleq\prod_{h=1}^{s}(\theta^{e_2f_1k}+\theta^{e_2f_1i_h})$.
$\prod_{h=1}^{s}(\theta^{e_2f_1k}+\theta^{e_2f_1i_h})$ and we denote it by $\omega$, i.e., $\omega=\prod_{h=1}^{s}(\theta^{e_2f_1k}+\theta^{e_2f_1i_h})$.
Since $e_1\mid e_2(r+1)$, then $\theta^{e_2f_1r}=\theta^{-e_2f_1}$ and
\[\begin{split}
	 \omega^r&=\prod_{h=1}^s \frac{\theta^{e_2f_1k}+\theta^{e_2f_1i_h}}{\theta^{e_2f_1(k+i_h)}}.\\
             &=\theta^{-e_2f_1(sk+I)}\omega,\\
\end{split}\]
where $I=\sum_{\mu=1}^s i_\mu$. Moreover, we can know that $\omega^{r-1}=\theta^{-e_2f_1(sk+I)}$ and
$$\omega=\theta^{-\frac{e_2(r+1)}{e_1}(sk+I)+c(r+1)},$$
where $c$ is some integer.
Considering that $\frac{e_2(r+1)}{e_1}$ is even, we have
\begin{equation}\label{eqqq3}
\eta(L_S(\gamma^{2j_\delta+1}\beta^k))=\eta(\omega)=1,
\end{equation}
which implies that $\eta(L_S(a))$ are the same for all $a\in S$ if $4\mid (s-1)(r+1)$.

Therefore, by Lemma \ref{rem1}, there exists a $q$-ary MDS self-dual code of length $n_1$ under the given conditions.

(2) By the proof of (1), we have
$$\eta(-L_S(\beta^{i_\lambda}\alpha^j))=\eta(-L_S(\gamma^{2j_\delta+1}\beta^k))=1.$$ It follows that $\eta(-L_S(a))=1$ for all $a\in S$.

Therefore, by Lemma \ref{lem2}, there exists a $q$-ary MDS self-dual code of length $n_1+1$ under the given conditions.

(3) It is easy to check that for $1\leq \lambda \leq s$ and $1\leq j\leq f_1$,
\[\begin{split}
	\eta(-\beta^{i_\lambda}\alpha^j L_S(\beta^{i_\lambda}\alpha^j))=(-1)^{\frac{(r+1)(s-1)}{2}}=1.
	\end{split}\]
For $1\leq \delta \leq t$ and $1\leq k\leq f_2$,
\[\begin{split}
	\eta(-\gamma^{2j_\delta+1}\beta^k L_S(\gamma^{2j_\delta+1}\beta^k))=\eta(\gamma^{f_2(2j_\delta+1)})=1
	\end{split}\]
and
\[\begin{split}
    \eta(-\prod_{j=1}^{f_1}\prod_{\mu=1}^{s}(\beta^{i_\mu}\alpha^j)\prod_{j=1}^{f_2}\prod_{\nu=1}^{t}(\gamma^{2j_\nu+1}\beta^j))=1.
\end{split}\]
And we note that $\eta(-1)=1$.

Therefore, by Lemma \ref{lem6}, there exists a $q$-ary MDS self-dual code of length $n_1+2$ under the given conditions. This completes the proof.
\end{proof}

Next, we present our second construction as follows, which can also be used to obtain three other classes of MDS self-dual codes.

\begin{theorem}\label{th 3}
Let $q=r^2$, where $r$ is an odd prime power.
Let $e_1\equiv 2^l({\rm mod}~2^{l+1})$ and $2^l\mid e_2$, where $l\geq 2$.
Suppose $2e_2\mid e_1(r+1)$ and $e_1\mid e_2(r-1)$. Let $1\leq s \leq D_1$, $1\leq t \leq D_2$ and $n_1=sf_1+tf_2$. The following statements hold.
\begin{itemize}
\item[(1)] If both $\frac{r+1}{2}(\frac{te_1}{e_2}+1)$ and $n_1$ are even, there exists a $q$-ary MDS self-dual code of length $n_1$.
\item[(2)] If both $\frac{r+1}{2}(\frac{te_1}{e_2}+t)$ and $\frac{(t-1)(r+1)}{2}$ are even, but $n_1$ is odd, there exists a $q$-ary MDS self-dual code of length $n_1+1$.
\item[(3)] If $\frac{r+1}{2}(\frac{te_1}{e_2}+t)$, $\frac{(t-1)(r+1)}{2}$ and $n_1$ are all even, there exists a $q$-ary MDS self-dual code of length $n_1+2$.
\end{itemize}
\end{theorem}

\begin{proof}
(1) Let $\theta$ be a primitive element of $\F_q$ and $\gamma=\theta^{\frac{e_1}{2}}$. Again, denote
\begin{equation}\label{eqqq4}
	M=\bigcup_{\mu=1}^{s}\beta^{i_\mu}A,\  N=\bigcup_{\nu=1}^{t}\gamma^{2j_\nu+1}B\  {\rm{and}}\ S=M\cup N,
\end{equation}
where $i_1, i_2,\dots, i_s$ are distinct module $D_1$ and $j_1, j_2,\dots, j_t$ are distinct module $D_2$.
By Lemma \ref{mod butong}, $\beta^{i_\mu}A\ (1\leq \mu \leq s)$, $\gamma^{2j_\nu+1}B\ (1\leq \nu\leq t)$ are distinct coests of $A$, $B$, respectively.
Taking a similar manner to Theorem \ref{th 2}, we have $M\cap N=\emptyset$. Let $f_S(x)=f_M(x)f_N(x)$ be the same as before. By Lemma \ref{lem5},
 \[\begin{split}
	 L_S(x)=f'_S(x)=f'_M(x)f_N(x)+f_M(x)f'_N(x).
\end{split}\]

On one hand, let $\beta^{i_\lambda}\alpha^j\in M$ for some $1\leq \lambda \leq s$ and $1\leq j\leq f_1$, then
\[\begin{split}
	 L_S(\beta^{i_\lambda}\alpha^j)&=L_M(\beta^{i_\lambda}\alpha^j)f_N(\beta^{i_\lambda}\alpha^j)\\
                &=f_1\theta^{e_2i_\lambda(f_1-1)}\theta^{-je_1}L_\mathbf{a}(\beta^{i_\lambda f_1})\prod_{h=1}^{t}(\theta^{e_1f_2j}-\theta^{\frac{e_1}{2}(2j_h+1)f_2}),\\
\end{split}\]
where $\mathbf{a}=(\beta^{i_1f_1},\beta^{i_2f_1},\cdots,\beta^{i_{s}f_1})$.
Note that $e_1\mid e_2(r-1)$, which implies that $\theta^{e_2f_1}\in \F_r^*$ and
$$\eta(f_1\theta^{e_2i_\lambda(f_1-1)}\theta^{-je_1}L_\mathbf{a}(\beta^{i_\lambda f_1}))=1.$$ Set $w\triangleq\prod_{h=1}^{t}(\theta^{e_1f_2j}-\theta^{\frac{e_1}{2}(2j_h+1)f_2})$
and note that $(\theta^{\frac{e_1f_2}{2}})^r=\theta^{\frac{e_1(r+1)(r^2-r)}{2e_2}}=\theta^{-\frac{e_1f_2}{2}}$ for $2e_2\mid e_1(r+1)$,
then $w^r$ can be calculated as
\[\begin{split}
      w^r&=\prod_{h=1}^{t}(\theta^{-e_1f_2j}-\theta^{-\frac{e_1}{2}(2j_h+1)f_2})\\
         &=\prod_{h=1}^{t}(-\frac{\theta^{e_1f_2j}-\theta^{\frac{e_1}{2}(2j_h+1)f_2}}{\theta^{\frac{e_1f_2}{2}(2j+2j_h+1)}})\\
         &=(-1)^t\theta^{-\frac{e_1f_2}{2}(2tj+t+2J)}w,
\end{split}\]
where $J=\sum_{h=1}^t j_h$. Clearly,
\[\begin{split}
       w^{r-1}=\theta^{\frac{r^2-1}{2}t-\frac{e_1f_2}{2}(2tj+t+2J)}
\end{split}\]
and
\[\begin{split}
       w=\theta^{\frac{r+1}{2}t-\frac{e_1(r+1)}{2e_2}(2tj+t+2J)+c(r+1)},
	\end{split}\]
where $c$ is some integer.
Moreover,
\[\begin{split}
       \eta(w)=(-1)^{\frac{r+1}{2}t-\frac{e_1(r+1)}{2e_2}t}=(-1)^{\frac{r+1}{2}(\frac{te_1}{e_2}+t)}
\end{split}\]
and
\begin{equation}\label{eqqq5}
\eta(L_S(\beta^{i_\lambda}\alpha^j))=\eta(w)=(-1)^{\frac{r+1}{2}(\frac{te_1}{e_2}+t)}.
\end{equation}

On the other hand, let $\gamma^{2j_\delta+1}\beta^k\in N$ for some $1\leq \delta \leq t$ and $1\leq k\leq f_2$, then
 \[\begin{split}
	 L_S(\gamma^{2j_\delta+1}\beta^k)&=L_N(\gamma^{2j_\delta+1}\beta^k)f_M(\gamma^{2j_\delta+1}\beta^k)\\
                &=f_2\theta^{\frac{e_1}{2}(2j_\delta+1)(f_2-1)}\theta^{-ke_2}\theta^{\frac{(t-1)e_1f_2}{2}}L_\mathbf{b}(\alpha^{j_\delta f_2})(-1)^s\prod_{h=1}^{s}(\theta^{e_2f_1k}+\theta^{e_2f_1i_h}),\\
\end{split}\]
where $\mathbf{b}=(\alpha^{j_1f_2},\alpha^{j_2f_2},\cdots,\alpha^{j_{t}f_2})$.
Similar to the proof of Theorem \ref{th 2} again, we can derive
\[\begin{split}
      \eta(f_2\theta^{\frac{e_1}{2}(2j_\delta+1)(f_2-1)}\theta^{-ke_2}\theta^{\frac{(t-1)e_1f_2}{2}}(-1)^s\prod_{h=1}^{s}(\theta^{e_2f_1k}+\theta^{e_2f_1i_h}))=1
\end{split}\]
from the fact $e_1\mid e_2(r-1)$. Hence, it is still sufficient to consider $L_\mathbf{b}(\alpha^{j_\delta f_2})$.
Since $(r-1)\mid e_1f_2$, then
\[\begin{split}
	        (L_{\mathbf{b}}(\alpha^{j_\delta f_2}))^r &= \prod_{\mu=1 ,\mu\neq \delta}^{t}(\alpha^{-j_\delta f_2}-\alpha^{-j_\mu f_2}) \\
                                              &= \prod_{\mu=1 ,\mu\neq \delta}^{t}\frac{\alpha^{j_\mu f_2}-\alpha^{j_\delta f_2}}{\alpha^{(j_\delta+j_\mu)f_2}}\\
                                              &= (-1)^{t-1}\alpha^{-j_\delta f_2(t-2)-f_2 J} L_{\mathbf{b}}(\alpha^{j_\delta f_2}),\\
\end{split}\]
where $J=\sum_{\mu=1}^t j_\mu$.
It follows that
\[\begin{split}
(L_\mathbf{b}(\alpha^{j_\delta f_2}))^{r-1}=\theta^{\frac{(t-1)(r^2-1)}{2}-e_1f_2[j_\delta (s-2)+J]}\\
	\end{split}\]
and
\[\begin{split}
      L_\mathbf{b}(\alpha^{j_\delta f_2})=\theta^{\frac{(t-1)(r+1)}{2}-\frac{e_1(r+1)}{e_2}[j_\delta (s-2)+J]+(r+1)c},
	\end{split}\]
where $c$ is some integer.
Therefore,
\begin{equation}\label{eqqq6}
\eta(L_S(\gamma^{2j_\delta+1}\beta^k))=(-1)^{\frac{(t-1)(r+1)}{2}},
\end{equation}
which yields $\frac{(t-1)(r+1)}{2}+\frac{r+1}{2}(\frac{te_1}{e_2}+t)$ is even if $\frac{r+1}{2}(\frac{te_1}{e_2}+1)$ is even and
it can further deduce that $\eta(L_S(a))$ are the same for all $a\in S$.

Therefore, by Lemma \ref{rem1}, there exists a $q$-ary MDS self-dual code of length $n_1$ under the given conditions.

(2) By the proof of (1), we have
$$\eta(-L_S(\beta^{i_\lambda}\alpha^j))=(-1)^{\frac{r+1}{2}(\frac{te_1}{e_2}+t)}=1$$
and
$$\eta(-L_S(\gamma^{2j_\delta+1}\beta^k))=(-1)^{\frac{(t-1)(r+1)}{2}}=1.$$
It follows that $\eta(-L_S(a))=1$ for all $a\in S$.

Therefore, by Lemma \ref{lem2}, there exists a $q$-ary MDS self-dual code of length $n_1+1$ under the given conditions.

(3) It is easy to check that for $1\leq \lambda \leq s$ and $1\leq j\leq f_1$,
\[\begin{split}
	\eta(-\beta^{i_\lambda}\alpha^j L_S(\beta^{i_\lambda}\alpha^j))=(-1)^{\frac{r+1}{2}(\frac{te_1}{e_2}+t)}=1.
	\end{split}\]
For $1\leq \delta \leq t$ and $1\leq k\leq f_2$,
\[\begin{split}
	\eta(-\gamma^{2j_\delta+1}\beta^k L_S(\gamma^{2j_\delta+1}\beta^k))=(-1)^{\frac{(t-1)(r+1)}{2}}=1
	\end{split}\]
and
\[\begin{split} \eta(-\prod_{j=1}^{f_1}\prod_{\mu=1}^{s}(\beta^{i_\mu}\alpha^j)\prod_{j=1}^{f_2}\prod_{\nu=1}^{t}(\gamma^{2j_\nu+1}\beta^j))=1.
	\end{split}\]
And we note that $\eta(-1)=1$.

Therefore, by Lemma \ref{lem6}, there exists a $q$-ary MDS self-dual code of length $n_1+2$ under the given conditions.  This completes the proof.
\end{proof}

\section{Comparisons and Examples}\label{sec4}

In this section,
we make a comparison of our results with the previous results
and give some examples.
In Table \ref{tab:4}, we list the ratios of $N$ and $\frac{q}{2}$
for some $q$, where $N$ is the number of possible lengths of
$q$-ary MDS self-dual codes in different constructions.

\begin{table}[h]
\caption{The ratios of $N$ and $\frac{q}{2}$}
\label{tab:4}
\begin{center}
\resizebox{\textwidth}{23mm}{
	\begin{tabular}{c |c |c| c| c| c}
		\hline
		 $r$  & $q$ & $N/(\frac{q}{2})$ of \cite{RefJ18} & $N/(\frac{q}{2})$ of \cite{RefJ20}& $N/(\frac{q}{2})$ of us and \cite{RefJ20} & \tabincell{c}{number of new lengths of \\us compared with \cite{RefJ20}} \\
		\hline
         149 &   22201  &  25\%  & 38.61\%  & 57.16\% & 2060\\
         151 &   22801  &  25\%  & 34.95\%  & 57.47\% & 2568 \\
         157 &   24649  &  25\%  & 34.95\%  & 57.10\% & 2731 \\
         163 &   26569  &  25\%  & 34.28\%  & 57.24\% & 3050 \\
         167 &   27889  &  25\%  & 34.27\%  & 57.36\% & 3219\\
		\hline
	\end{tabular}}
    \begin{tablenotes}
     \footnotesize
    \item Note: $N$ is the number of possible lengths of $q$-ary MDS self-dual codes.
    \end{tablenotes}
\end{center}
\end{table}

According to Table \ref{tab:4},
the constructions in \cite{RefJ18} can contribute 25\% of all possible $q$-ary MDS self-dual codes.
And in \cite{RefJ20}, the constructions can contribute more than 34\% of all possible $q$-ary MDS self-dual codes.
Inspired by the idea in \cite{RefJ20}, we construct six new classes of MDS self-dual codes in Theorems \ref{th 2}
and \ref{th 3}. After a detailed computer search, we find that some possible lengths in our results cannot be obtained
in \cite{RefJ20}. In the last column of Table \ref{tab:4}, we list the number of possible lengths in our constructions but
not in the constructions in \cite{RefJ20}. Of course, there are also some lengths that can be obtained from \cite{RefJ20},
but not from our constructions. However, if we consider with our results and conclusions in \cite{RefJ20} together,
we can find that $N/(\frac{q}{2})$ is generally more than 57\%, which is the largest ratio as far as we know.

More intuitively, we give some new concrete examples of MDS self-dual codes in Examples \ref{exam1} and \ref{exam2}.

\begin{example}\label{exam1}
Let $q=149^2$, by Theorem \ref{th 2},
there exists an MDS self-dual code of length $n=7504,\ 8180,\ 4944,\ 6172,\ 9018 \cdots$.
Calculated by Magma, these lengths are new and cannot be obtained from previous constructions listed in Table \ref{tab:2}.
\end{example}

\begin{example}\label{exam2}
Let $q=151^2$, by Theorem \ref{th 3},
there exists an MDS self-dual code of length $n=7148,\ 9592,\ 6616,\ 10040,\ 8288 \cdots$.
Calculated by Magma, these lengths are new and cannot be obtained from previous constructions listed in Table \ref{tab:2}.
\end{example}

\section{Constructions of MDS self-orthogonal codes and MDS almost self-dual codes}\label{sec5}

In this section, we apply Lemmas \ref{lem self-orth 1} and \ref{lem self-orth 2}
to construct some new MDS self-orthogonal codes and MDS almost self-dual codes.

\begin{theorem}\label{th 4}
Let $q=r^2$, where $r$ is an odd prime power.
Let $e_1\equiv 2^{l}({\rm mod}~2^{l+1})$ and $2^{l}\mid e_2$, where $l\geq 2$.
Suppose $2e_2\mid e_1(r-1)$ and $e_1\mid e_2(r+1)$.
Let $1\leq s \leq D_1$, $1\leq t \leq D_2$, $4\mid (s-1)(r+1)$ and $n_1=sf_1+tf_2$. The following statements hold.
\begin{itemize}
\item[(1)] If $n_1$ is even, for $1\leq k\leq \frac{n_1}{2}-1$, there exists a $q$-ary $[n_1,k]$ MDS self-orthogonal code.
\item[(2)] If $n_1$ is even, there exists a $q$-ary MDS almost self-dual code of length $n_1+1$.
\item[(3)] If $n_1+1$ is even, for $1\leq k\leq \frac{n_1+1}{2}-1$, there exists a $q$-ary $[n_1+1,k]$ MDS self-orthogonal code.
\end{itemize}
\end{theorem}

\begin{proof}
(1) Let $S$ be defined as Eq. (\ref{eqqq1}). Naturally, Eqs. (\ref{eqqq2}) and (\ref{eqqq3}) follow.
Set $\omega(x)=x$, then
for $1\leq \lambda \leq s$ and $1\leq j\leq f_1$,
\[\begin{split}
	\eta(\omega(\beta^{i_\lambda}\alpha^j)L_S(\beta^{i_\lambda}\alpha^j))=(-1)^{\frac{(r+1)(s-1)}{2}}=1.
	\end{split}\]
For $1\leq \delta \leq t$ and $1\leq k\leq f_2$,
\[\begin{split}
	\eta(\omega(\gamma^{2j_\delta+1}\beta^k)L_S(\gamma^{2j_\delta+1}\beta^k))=\eta(\gamma^{f_2(2j_\delta+1)})=1.
	\end{split}\]
It follows that there exists a nonzero polynomial $\omega(x)$ such that $\eta(\omega(a)L_S(a))=1$ for all $a\in S$.

Therefore, by Lemma \ref{lem self-orth 1}, there exists a $q$-ary $[n_1,k]$ MDS self-orthogonal code under the given conditions.

(2) By the proof of (1), set $\omega(x)=-x$ here. Similarly, we can prove that  $\eta(\omega(a)L_S(a))=1$ for all $a\in S$.
Therefore, by Lemma \ref{lem self-orth 2}, there exists a $q$-ary $[n_1+1,\frac{n_1}{2}]$ MDS self-orthogonal code. Moreover, note that
the dual code is a $q$-ary $[n_1+1,\frac{n_1}{2}+1]$ MDS code, then according to the definition of almost self-dual codes, each $q$-ary $[n_1+1,\frac{n_1}{2}]$
MDS self-orthogonal code is actually a $q$-ary MDS almost self-dual code.
Therefore, there exists a $q$-ary MDS almost self-dual codes of length $n_1+1$ under the given conditions.

(3) Denote $\tilde{S}=S\cup\{0\}$. Set $\omega(x)=\theta^{\frac{e_1}{2}}$, then
for $1\leq \lambda \leq s$ and $1\leq j\leq f_1$,
\[\begin{split}
	\eta(\omega(\beta^{i_\lambda}\alpha^j)L_{\tilde{S}}(\beta^{i_\lambda}\alpha^j))=\eta(-\beta^{i_\lambda}\alpha^j\omega(\beta^{i_\lambda}\alpha^j) L_{S}(\beta^{i_\lambda}\alpha^j))=(-1)^{\frac{(r+1)(s-1)}{2}}=1.
	\end{split}\]
For $1\leq \delta \leq t$ and $1\leq k\leq f_2$,
\[\begin{split}
	\eta(\omega(\gamma^{2j_\delta+1}\beta^k)L_{\tilde{S}}(\gamma^{2j_\delta+1}\beta^k))=\eta(-\gamma^{2j_\delta+1}\beta^k\omega(\gamma^{2j_\delta+1}\beta^k)L_S(\gamma^{2j_\delta+1}\beta^k))=1
	\end{split}\]
and
\[\begin{split} \eta(\omega(0)L_{\tilde{S}}(0))=\eta(\prod_{j=1}^{f_1}\prod_{\mu=1}^{s}(\beta^{i_\mu}\alpha^j)\prod_{j=1}^{f_2}\prod_{\nu=1}^{t}(\gamma^{2j_\nu+1}\beta^j))=1.
	\end{split}\]
It follows that there exists a nonzero polynomial $\omega(x)$ such that $\eta(\omega(a)L_{\tilde{S}}(a))=1$ for all $a\in \tilde{S}$.
Therefore, by Lemma \ref{lem self-orth 1}, there exists a $q$-ary $[n_1+1,k]$ MDS self-orthogonal code under the given conditions. This completes the proof.
\end{proof}

\begin{theorem} \label{th 5}
Let $q=r^2$, where $r$ is an odd prime power.
Let $e_1\equiv 2^l({\rm mod}~2^{l+1})$ and $2^{l}\mid e_2$, where $l\geq 2$.
Suppose $2e_2\mid e_1(r+1)$ and $e_1\mid e_2(r-1)$. Let $1\leq s \leq D_1$, $1\leq t \leq D_2$ and $n_1=sf_1+tf_2$. The following statements hold.
\begin{itemize}
\item[(1)] If  both $\frac{r+1}{2}(\frac{te_1}{e_2}+1)$ and $n_1$ are even, for $1\leq k\leq \frac{n_1}{2}-1$, there exists a $q$-ary $[n_1,k]$-MDS self-orthogonal code.
\item[(2)] If  $\frac{t(r+1)e_1}{2e_2}$, $\frac{r+1}{2}$ and $n_1$ are all even, there exists a $q$-ary MDS almost self-dual code of length $n_1+1$.
\item[(3)] If  both $\frac{r+1}{2}(\frac{te_1}{e_2}+t)$ and $\frac{(t-1)(r+1)}{2}$ are even, but $n_1$ is odd, for $1\leq k\leq \frac{n_1+1}{2}-1$, there exists a $q$-ary $[n_1+1,k]$ MDS self-orthogonal code.
\end{itemize}
\end{theorem}

\begin{proof}
(1) Let $S$ be defined as Eq. (\ref{eqqq4}). Naturally, Eqs. (\ref{eqqq5}) and (\ref{eqqq6}) hold. We divide our proof
into the following two parts according to whether 4 is divisible by $r+1$.

$\bullet$ $\textbf{Case 1:}$ $4\mid (r+1)$.

Set $\omega(x)=x$, then for $1\leq \lambda \leq s$ and $1\leq j\leq f_1$,
\[\begin{split}
	\eta(\omega(\beta^{i_\lambda}\alpha^j)L_S(\beta^{i_\lambda}\alpha^j))=(-1)^{\frac{r+1}{2}(\frac{te_1}{e_2}+t)}=1.
	\end{split}\]
For $1\leq \delta \leq t$ and $1\leq k\leq f_2$,
\[\begin{split}
	\eta(\omega(\gamma^{2j_\delta+1}\beta^k)L_S(\gamma^{2j_\delta+1}\beta^k))=(-1)^{\frac{(t-1)(r+1)}{2}}=1.
	\end{split}\]

$\bullet$ $\textbf{Case 2:}$ $4\nmid (r+1)$.

Set $\omega(x)=\theta^{t\frac{r+1}{2}+1}$, then for $1\leq \lambda \leq s$ and $1\leq j\leq f_1$,
\[\begin{split}
	\eta(\omega(\beta^{i_\lambda}\alpha^j)L_S(\beta^{i_\lambda}\alpha^j))=(-1)^{\frac{te_1(r+1)}{2e_2}+1}=1.
	\end{split}\]
For $1\leq \delta \leq t$ and $1\leq k\leq f_2$,
\[\begin{split}
	\eta(\omega(\gamma^{2j_\delta+1}\beta^k)L_S(\gamma^{2j_\delta+1}\beta^k))=(-1)^{\frac{(r+1)}{2}+1}=1.
	\end{split}\]
It follows that there exists a nonzero polynomial $\omega(x)$ such that $\eta(\omega(a)L_S(a))=1$ for all $a\in S$
whether 4 is divisible by $r+1$ or not.

Therefore, by Lemma \ref{lem self-orth 1}, there exists a $q$-ary $[n_1,k]$ MDS self-orthogonal code under the given conditions.

(2) By the proof of (1), set $\omega(x)=-x$. Similarly, we can prove that  $\eta(\omega(a)L_S(a))=1$ for all $a\in S$.
Therefore, by Lemma \ref{lem self-orth 2},
there exists a $q$-ary $[n_1+1,\frac{n_1}{2}]$ MDS self-orthogonal code. Moreover, note that
the dual code is a $q$-ary $[n_1+1,\frac{n_1}{2}+1]$ MDS code, then according to the definition of almost self-dual codes, each $q$-ary $[n_1+1,\frac{n_1}{2}]$
MDS self-orthogonal code is actually a $q$-ary MDS almost self-dual code.
Therefore, there exists a $q$-ary MDS almost self-dual codes of length $n_1+1$ under the given conditions.

(3) Denote $\tilde{S}=S\cup\{0\}$. Set $\omega(x)=\theta^{\frac{e_1}{2}}$, then for $1\leq \lambda \leq s$ and $1\leq j\leq f_1$,
\[\begin{split}
	\eta(\omega(\beta^{i_\lambda}\alpha^j)L_{\tilde{S}}(\beta^{i_\lambda}\alpha^j))=\eta(-\beta^{i_\lambda}\alpha^j\omega(\beta^{i_\lambda}\alpha^j) L_{S}(\beta^{i_\lambda}\alpha^j))=(-1)^{\frac{r+1}{2}(\frac{te_1}{e_2}+t)}=1.
	\end{split}\]
For $1\leq \delta \leq t$ and $1\leq k\leq f_2$,
\[\begin{split}
	\eta(\omega(\gamma^{2j_\delta+1}\beta^k)L_{\tilde{S}}(\gamma^{2j_\delta+1}\beta^k))\!=\!\eta(-\gamma^{2j_\delta+1}\beta^kL_S(\gamma^{2j_\delta+1}\beta^k))\!=\!(-1)^{\frac{(t-1)(r+1)}{2}}=1
	\end{split}\]
and
\[\begin{split} \eta(\omega(0)L_{\tilde{S}}(0))=\eta(\prod_{j=1}^{f_1}\prod_{\mu=1}^{s}(\beta^{i_\mu}\alpha^j)\prod_{j=1}^{f_2}\prod_{\nu=1}^{t}(\gamma^{2j_\nu+1}\beta^j))=1.
	\end{split}\]
It follows that there exists a nonzero polynomial $\omega(x)$ such that $\eta(\omega(a)L_{\tilde{S}}(a))=1$ for all $a\in \tilde{S}$.

Therefore, by Lemma \ref{lem self-orth 1}, there exists a $q$-ary $[n_1,k]$ MDS self-orthogonal code under the given conditions.  This completes the proof.
\end{proof}

\section{Summary and concluding remarks}\label{sec6}
In this paper, we develop the construction methods in \cite{RefJ20} and construct six new classes of MDS self-dual
codes over $\F_q$ via (extended) GRS codes, where $q=r^2$ and $r$ is an odd prime power (See Theorems \ref{th 2} and \ref{th 3}).
Combined with \cite{RefJ20}, the ratio $N/(\frac{q}{2})$ is generally more than 57\%, which is the largest known ratio as far as we know
(See Table \ref{tab:4}). Moreover, by using the same evaluation sets, some new families of MDS self-orthogonal codes and MDS almost
self-dual codes are also given (See Theorems \ref{th 4} and \ref{th 5}). Our work enriches the study of MDS self-dual codes. 
The future work is to find the construction of MDS self-dual codes which can take up a larger proportion,
and even to completely slove the problem of the construction of MDS self-dual codes for all possible length.

\section*{Acknowledgments}
This research was supported by the National Natural Science Foundation of China (Nos.U21A20428 and 12171134).

\end{document}